\documentclass[12pt,letterpaper]{article}

\usepackage[utf8]{inputenc}
\usepackage{graphicx}
\usepackage{xcolor}

\usepackage{amsmath}
\usepackage{amssymb}
\usepackage{amsthm}
\usepackage{mathtools}
\usepackage{cuted}
\usepackage{lineno}
\usepackage{caption}
\usepackage{dsfont}
\usepackage{subfig}
\usepackage{hyperref}
\usepackage{enumitem}
\usepackage{algorithm}
\usepackage{algpseudocode}

\algtext*{EndFor}
\algtext*{EndIf}
\algtext*{EndFunction}

\newtheorem{theorem}{Theorem}
\newtheorem{corollary}{Corollary}
\newtheorem{lemma}{Lemma}

\usepackage[letterpaper]{geometry}
\title{Weighted Optimization with Multiple Axis-Aligned Rectangles}
\author{
José Fernández Goycoolea\thanks{Departamento de Matem\'atica y F\'isica, Universidad de Magallanes, Avenida Bulnes 01855, Punta Arenas, Chile, {\tt jose.fernandezg@umag.cl}, https://orcid.org/0000-0001-5349-4348.} 
\and
Luis H. Herrera\thanks{Departamento de Inform\'atica y Computaci\'on, Universidad Tecnol\'ogica Metropolitana, Jos\'e Pedro Alessandri 1242, \~{N}u\~{n}oa, Santiago de Chile 7800002, Región Metropolitana, Chile, 
{\tt luis.herrerab@utem.cl}, https://orcid.org/0000-0001-7338-7611.}
\and
Pablo P\'erez-Lantero\thanks{Departamento de Matem\'atica y Ciencia de la Computaci\'on, Universidad de Santiago de Chile (USACH), Las Sophoras 173, Santiago de Chile, Región Metropolitana, Chile, {\tt pablo.perez.l@usach.cl}, https://orcid.org/0000-0002-8703-8970.
}
\and
Carlos Seara\thanks{Departament de Matem\`atiques, Universitat Polit\`ecnica de Catalunya, Jordi Girona 1, 08034 Barcelona, Spain, {\tt carlos.seara@upc.edu}, https://orcid.org/0000-0002-0095-1725.
}
}

\begin{document}
\maketitle

\begin{abstract}
Let $P$ be a set of $n$ weighted points in the plane. We study optimization with axis-aligned rectangles under objectives determined by the incidence pattern of each point with the selected rectangles. As a central case, we consider two possibly overlapping rectangles maximizing the total weight of the points in their symmetric difference. We give a direct algorithm based on a generalized Maximum Consecutive Subsequence tree running in $O(n^4\log n)$ time and $O(n)$ space, and improve the running time to $O(n^4)$ via a reduction to Weighted Depth in $\mathbb{R}^8$. The resulting parameter-space instance has the signed coefficient pattern $(w,w,-2w)$ induced by symmetric difference. When the two factor boxes are allowed to be arbitrary and the associated weights may be signed, this coefficient structure is linearly equivalent, up to an additive constant, to general Weighted Depth. More generally, for any fixed $k$ and fixed Boolean incidence objective $h:\{0,1\}^k\to\mathbb{R}$, a Möbius expansion over the Boolean lattice yields an $O(n^{2k})$-time algorithm through Weighted Depth in $\mathbb{R}^{4k}$. This strictly extends depth-based objectives by allowing the value of a point to depend on which rectangles cover it. Symmetric difference, union, intersection, exact-depth and threshold coverage, and overlap rewards or penalties are all special cases.
\end{abstract}



\section{Introduction}

Covering weighted point sets with simple geometric regions is a classical problem in computational geometry. Given a finite set $P$ of $n$ points in the plane and a weight function $\omega:P\to\mathbb{R}$, the goal is to choose one or more geometric regions so as to maximize the total weight of the points selected by them. Variants of this problem have been studied for several geometric objects~\cite{baral2017maximum,CABELLO2013203,CABELLO2008195,chazellegg}, with applications and connections to data analysis, clustering, discrepancy theory, and machine learning~\cite{DOBKIN1996453,Erick}.

A basic example is the \textsc{Maximum Weighted Box} problem, where the region is an axis-aligned rectangle. Cortés et al.~\cite{CortesDPSUV09} gave an $O(n^2\log n)$-time algorithm in the plane by reducing the problem to one-dimensional subproblems handled with a \emph{Maximum Consecutive Subsequence tree} (MCS-tree). Barbay et al.~\cite{BARBAY2014437} later obtained an $O(n^2)$-time worst-case bound, together with adaptive algorithms that are faster on suitable instances. Related weighted-selection problems include maximum-weight islands~\cite{BAUTISTASANTIAGO2011246,doi:10.1137/19M1303010} and bichromatic separation problems with strips, boxes, and guillotine partitions~\cite{BEREG201365,DIAZBANEZ20171,DIAZBANEZ201780,FERNANDEZGOYCOOLEA2024100503}.

We study the natural two-box extension: \emph{Given two possibly overlapping axis-aligned rectangles $A$ and $B$, we want to maximize the total weight of the points in their symmetric difference $A\triangle B$}. Thus, a point contributes when it is covered by exactly one of the two rectangles, while a point covered by both contributes nothing. Problems involving two boxes have been considered under other objectives, such as minimum-area enclosure and point-set covering~\cite{arkin2006algorithms,becker1996enclosing,bespamyatnikh2000covering}; the symmetric-difference objective leads to a different interaction between the two rectangles.

A particularly relevant predecessor is the work of Dobkin et al.~\cite{DOBKIN1996453}. In the bichromatic discrepancy setting they considered the union of two disjoint axis-aligned rectangles and obtained an $O(n^2\log n)$-time algorithm; they also developed dynamic machinery for unions of a fixed number of intervals. Their two-rectangle result relies on disjointness. Our central problem instead permits arbitrary overlap, and the general framework below allows the contribution of a point to depend on its complete incidence pattern with the selected rectangles.

After establishing the direct symmetric-difference bound, we show that the fixed zero-one sector coefficients used by the algorithm can be replaced by arbitrary real values determined by the complete incidence pattern of each point with the labeled boxes. This leads to a Boolean incidence framework that can distinguish which boxes contain a point, rather than only how many boxes contain it. Coverage-depth objectives are recovered as the symmetric special cases of this framework.

We use two complementary approaches. The first works directly in the plane. After fixing the horizontal boundaries of the boxes, a sector decomposition reduces the remaining optimization to a one-dimensional problem maintained by a generalized MCS-tree. For two boxes this gives $O(n^4\log n)$ time and $O(n)$ space. The second algorithm moves to parameter space: two planar rectangles are represented by their eight boundary coordinates, and each input point gives rise to a constant number of boxes in $\mathbb{R}^8$. The use of rectangle-boundary coordinates as parameter-space variables follows the standard connection between \textsc{Maximum-Weight Box} and \textsc{Weighted Depth} used by Barbay et al.~\cite{BARBAY2014437}; the new ingredient here is the encoding of the joint incidence pattern of two labeled rectangles, including the signed interaction induced by symmetric difference. The resulting instance is a \textsc{Weighted Depth} problem, and Chan's fixed-dimensional framework~\cite{Chan2013Klee}, together with the real-weighted formulations used for maximum-weight boxes~\cite{BackursDikkalaTzamos2016,BARBAY2014437}, yields an $O(n^4)$-time algorithm. More generally, $k$ rectangles are represented in dimension $4k$. A M\"obius expansion writes every Boolean incidence objective as a linear combination of simultaneous membership events in subfamilies of the selected boxes, producing an $O(n^{2k})$-time algorithm for every fixed $k$.

Our main contributions are the following.
\begin{enumerate}[label=(\roman*)]
    \item We generalize the MCS-tree to maintain the best three consecutive $x$-blocks while one rectangle boundary is swept. This yields a direct $O(n^4\log n)$-time and $O(n)$-space algorithm for \textsc{Maximum-Weight Two Boxes Symmetric Difference}.
    \item We show that the same direct framework handles every fixed Boolean incidence objective by changing only the sector coefficients. In particular, the algorithm can distinguish two sectors having the same coverage depth but covered by different labeled boxes.
    \item Using the standard boundary-coordinate parameterization, we encode the two-box problem as \textsc{Weighted Depth} on $O(n)$ boxes in $\mathbb{R}^8$, improving the worst-case running time to $O(n^4)$. The same reduction handles every fixed two-box Boolean incidence objective. The essential interaction for symmetric difference is the signed coefficient pattern $(w,w,-2w)$. With arbitrary factor boxes, this coefficient structure is linearly equivalent, up to an additive constant, to general Weighted Depth. Hence the coefficient pattern alone cannot yield a truly sub-$n^4$ algorithm without also improving the corresponding general Weighted Depth bound in dimension eight.
    \item For fixed $k$, the generalized MCS-tree uses $O(k^2n)$ space and has $O(k^3)$ merge cost. A second reduction, based on M\"obius inversion on the Boolean lattice, produces $(2^k-1)n=O(n)$ boxes in $\mathbb{R}^{4k}$ and gives an $O(n^{2k})$-time algorithm for every fixed Boolean incidence objective.
\end{enumerate}

\textbf{Paper outline.} Section~\ref{sec:MWTB} develops the direct two-box algorithm and the generalized MCS-tree. Section~\ref{sec:RESDISC} gives the main bounds: after proving the direct symmetric-difference result, Section~\ref{sec:boolean_objectives} generalizes the fixed sector coefficients to Boolean incidence objectives, and Section~\ref{sec:weighted_depth} gives the parameter-space reduction. Section~\ref{sec:k_boxes} extends both approaches to a fixed number $k$ of boxes and analyzes the computational complexity. Applications and variants are discussed in Section~\ref{sec:applications}, followed by conclusions in Section~\ref{sec:Conclu}.

\subsection{The MCS-tree}

We briefly recall the MCS-tree of Cortés et al.~\cite{CortesDPSUV09}. It is a balanced binary tree that maintains a maximum-weight consecutive subsequence of a weighted sequence under point updates.

Let $X=(x_1,\ldots,x_n)$ be a sequence. A \emph{block} of $X$ is a consecutive subsequence; we also allow the empty block. Two blocks are \emph{contiguous} if their concatenation is again a block of $X$. For a block $X_{i,j}=(x_i,x_{i+1},\ldots,x_j)$ and a function $f:\{x_1,\ldots,x_n\}\to\mathbb{R}$, we write
\[
    f(X_{i,j})=\sum_{k=i}^j f(x_k).
\]

The leaves of the MCS-tree correspond to the elements of $X$, and each internal node $U$ represents a canonical block. If $U_1$ and $U_2$ are its children, then $U=U_1U_2$. For every node $U$ we store four values,
\[
    \mathcal{A}(U)=\{S(U),L(U),M(U),R(U)\},
\]
where $M(U)$ is the maximum weight of a block of $U$; $L(U)$ and $R(U)$ are the maximum weights of a block touching the left and right endpoint of $U$, respectively, with the empty block allowed; and $S(U)$ is the total weight of $U$. For any block $V$ of $X$, write
\[
    \omega(V)=\sum_{x\in V}\omega(x),
\]
with $\omega(\emptyset)=0$. More explicitly, if $U=(x_l,\ldots,x_r)$,
\begin{enumerate}
    \item $M(U)=\max_V\omega(V)$, where $V$ is a block of $U$;
    \item $L(U)=\max_V\omega(V)$, where $V$ is empty or contains $x_l$;
    \item $R(U)=\max_V\omega(V)$, where $V$ is empty or contains $x_r$;
    \item $S(U)=\omega(U)$.
\end{enumerate}
The optimum for the whole sequence is therefore $M(X)$ at the root.

The four values at an internal node are obtained in constant time from its children:
\begin{equation}\label{eq:A(U)recursions}
    \begin{split}
        S(U) &= S(U_1) + S(U_2),                          \\
        L(U) &= \max\{L(U_1),~ S(U_1) + L(U_2)\},         \\
        M(U) &= \max\{M(U_1),~ R(U_1) + L(U_2),~ M(U_2)\},\\
        R(U) &= \max\{R(U_1) + S(U_2),~ R(U_2)\}.
    \end{split}
\end{equation}
For a leaf $U=(x)$, we set $S(U)=\omega(x)$ and
$L(U)=M(U)=R(U)=\max\{0,\omega(x)\}$. All four values are $0$ for an empty block. Figure~\ref{fig:example_MCStree} illustrates the construction on a sequence of length $16$.

    \begin{figure}[h]
        \centering
        \includegraphics[width=\linewidth,keepaspectratio]{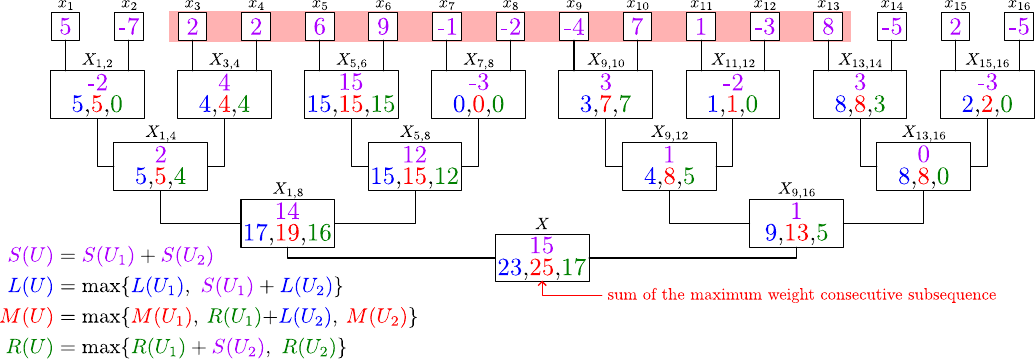}
        \caption{Maximum-weight consecutive subsequence computed with the MCS-tree.}
        \label{fig:example_MCStree}
    \end{figure}

The tree is built in $O(n)$ time. If the weight of one element changes, only the nodes on its root-to-leaf path must be recomputed, giving an $O(\log n)$ update time. If the maximizing block itself is needed, its endpoint indices can be stored together with the four values; see~\cite{CortesDPSUV09}.

\section{Maximum-Weight Two Boxes Symmetric Difference}\label{sec:MWTB}

Let $P=\{p_1,\ldots,p_n\}$ be a set of weighted points in the plane, with weight function $\omega:P\to\mathbb{R}$. Throughout the paper, rectangles and higher-dimensional boxes are closed. Since the objective depends only on the incidences with the finite set $P$, every realizable incidence pattern has a representative whose boundary avoids the input points: a side passing through points of $P$ can be displaced by a sufficiently small amount into the adjacent coordinate gap without crossing any point that would change the pattern.

For the sweep description we first assume that no two points share an
$x$- or $y$-coordinate. This assumption is only notational. Points
having the same coordinate can be processed simultaneously at the
corresponding sweep event. They may also be ordered symbolically,
provided that no rectangle boundary is allowed to separate points
having the same coordinate. In the MCS-tree, points sharing an
$x$-coordinate may equivalently be grouped into one leaf, whose
$f_j$-value is the sum of their contributions. Thus, ties do not
change the feasible incidence patterns or the objective value.

We seek two axis-aligned rectangles $A$ and $B$ that maximize the total weight of the points in $A\triangle B$. The rectangles may overlap and need not be distinct; in particular, $A=B$ is allowed. We also allow rectangles disjoint from $P$. Our direct algorithm is based on a generalization of the MCS-tree introduced in the previous section.

Up to rotations and reflections, Figure~\ref{fig:cases} shows the five relative placements used by the direct two-box algorithm: containment intersection, corner intersection, cross intersection, side intersection, and disjointness. We label the rectangles $A$ and $B$ only to describe the configurations.

    \begin{figure}[h]
    	\centering
    	\subfloat[]{
    		\includegraphics[scale=0.7,page=1]{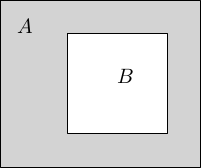}
    		\label{fig:contain}
    	}~~
    	\subfloat[]{
    		\includegraphics[scale=0.7,page=2]{img.pdf}
    		\label{fig:corner}
    	}~~
    	\subfloat[]{
    		\includegraphics[scale=0.7,page=3]{img.pdf}
    		\label{fig:cross}
    	}~~
    	\subfloat[]{
    		\includegraphics[scale=0.7,page=4]{img.pdf}
    		\label{fig:side}
    	}~~
    	\subfloat[]{
    		\includegraphics[scale=0.7,page=5]{img.pdf}
    		\label{fig:disjoint}
    	}
    
            \bigskip
    	\caption{
    			(a) Containment intersection.
    			(b) Corner intersection.
    			(c) Cross intersection.
    			(d) Side intersection. 
    			(e) Disjoint.
    	}
    	\label{fig:cases}
    \end{figure} 
Because the objective is symmetric in $A$ and $B$, the labels in Figure~\ref{fig:cases} are interchangeable; panel~(d) fixes one orientation only for illustration. We treat the configurations separately, but the same optimization framework applies to all of them. We describe that framework next.

\subsection{Optimization framework}
    
Fix four horizontal lines $y=\ell_i$, $i=0,\ldots,3$, with $\ell_0\geq\ell_1\geq\ell_2\geq\ell_3$, and assume that no input point lies on one of them. For $i=1,2,3$, let $\varphi_i:P\to\mathbb{R}$ indicate the strip between $\ell_{i-1}$ and $\ell_i$:
    \begin{equation}\label{eq:horizontal_strips}
        \varphi_i(p)=\left\{\begin{array}{ll}
        1 & \text{if }\ell_{i}<p_y<\ell_{i-1},\\
        0 & \text{otherwise.}
        \end{array}\right.
    \end{equation}
    
For each vertical block $j=1,2,3$, define
    \begin{equation}\label{eq:vertcal_strips}
        f_j(p)=\sum_{i=1}^3 a_{i,j}~\varphi_i(p)~\omega(p),
    \end{equation}
where $a_{i,j}$ is the entry in row $i$ and column $j$ of a $3\times3$ zero-one \emph{sector activation matrix}. Now sort the points by increasing $x$-coordinate and write the resulting sequence as $X=(x_1,\ldots,x_n)$. For fixed horizontal lines, the remaining task is one-dimensional: choose three consecutive labeled blocks $X_1,X_2,X_3$ maximizing $\sum_{j=1}^3 f_j(X_j)$. Figure~\ref{fig:framework} illustrates this reduction for the corner and side-intersection configurations of Figure~\ref{fig:cases}.

    \begin{figure}[h]
        \centering
        \includegraphics{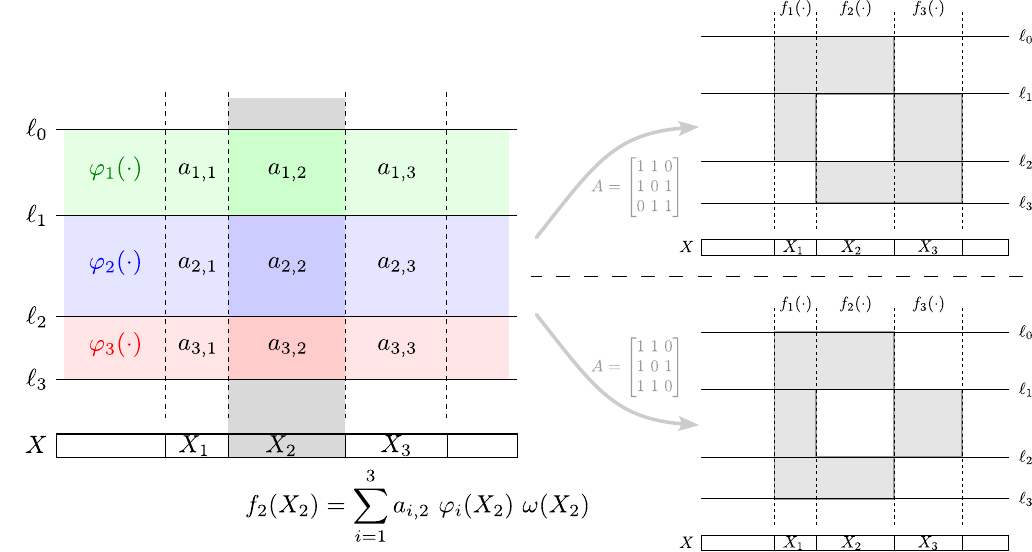}
        \bigskip
        \caption{Example of the optimization framework.}
    	\label{fig:framework}
    \end{figure}

The generalized MCS-tree described below solves this one-dimensional problem in $O(n)$ time once the lines are fixed. After the tree has been built, changing the contribution of one point requires only an $O(\log n)$ update.

To optimize over all placements, we fix three horizontal boundaries and sweep the fourth. There are $O(n^3)$ relevant choices for the fixed boundaries and $O(n)$ events in each sweep. Building the tree costs $O(n)$, while each event costs $O(\log n)$. This gives $O(n^4\log n)$ time for the two-box problem.

\subsection{Generalized MCS-tree}

The generalized MCS-tree uses the same balanced-tree structure as the ordinary MCS-tree. Each node $U$ represents a canonical block of $X=(x_1,\ldots,x_n)$ and stores a collection $\mathcal{A}(U)$ of partial optima. The labeled blocks $V_j$ may be empty and contribute zero. An ordered family $V_s,\ldots,V_e$ is called a \emph{consecutive labeled family} if, after deleting the empty blocks, the remaining blocks are pairwise disjoint, appear from left to right in increasing label order, and their concatenation is a block of $U$. Thus empty labels preserve the order but do not create gaps between two nonempty blocks.
    \begin{enumerate}
    \item $\displaystyle M(U)=\max_{V_1,V_2,V_3}\sum_{j=1}^3f_j(V_j)$ where $V_1,V_2,V_3$ form a consecutive labeled family in $U$. 
    \end{enumerate}
Thus, the optimum for fixed horizontal lines is $M(X)$ at the root. To combine solutions across two child nodes, we also store
$\mathcal{A}(U)=\{S_{s,e}(U),L_s(U),M(U),R_e(U):1\leq s\leq e\leq3\}$,
for a total of $13$ values. They record the same type of optimum under endpoint and label-range restrictions. If $U=(x_l,\ldots,x_r)$, the definitions are:
    \begin{enumerate}[resume]
        \item $L_s(U)$ forces the inclusion of the leftmost element or is empty:
        
        $\displaystyle L_s(U)=\max_{V_s,\ldots,V_3}\sum_{j=s}^3f_j(V_j),$ where $V_s,\ldots,V_3$ form a consecutive labeled family in $U$ for $1\leq s\leq3$, with $x_l$ contained in one of them or all of them being empty.
        \item $R_e(U)$ forces the inclusion of the rightmost element or is empty: 
        
        $\displaystyle  R_e(U)=\max_{V_1,\ldots,V_e}\sum_{j=1}^ef_j(V_j),$ where $V_1,\ldots,V_e$ form a consecutive labeled family in $U$ for $1\leq e\leq3$, with $x_r$ contained in one of them or all of them being empty.
        \item $S_{s,e}(U)$ forces the inclusion of both the leftmost and the rightmost elements: 
        
        $\displaystyle S_{s,e}(U)=\max_{V_s,\ldots,V_e}\sum_{j=s}^ef_j(V_j),$ where $V_s,\ldots,V_e$ form a consecutive labeled family in $U$ for $1\leq s\leq e\leq3$, with $x_l$ and $x_r$ contained in some of them.
    \end{enumerate}

    \begin{figure}[ht]
        \centering
        \includegraphics[scale=1,page=1]{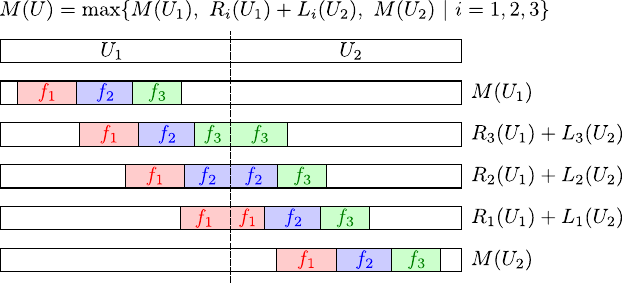}

        \medskip\hrule
        
        \includegraphics[scale=1,page=2]{box_recursion.pdf}
        \caption{Recursive computation cases of $M(U)$ (top) and $R_2(U)$ (bottom).}
        \label{fig:M(U)recursion}
    \end{figure}

For two boxes, all values at an internal node $U=U_1U_2$ can be computed in constant time from the values stored at its children. At a leaf $U=(x)$, we initialize the states directly from $f_1,f_2,f_3$:
    \[
    \begin{aligned}
       S_{s,e}(U) &= \max_{s\leq j\leq e} f_j(x),\\
       L_s(U) &= \max\left\{0,\max_{s\leq j\leq 3} f_j(x)\right\},\\
       R_e(U) &= \max\left\{0,\max_{1\leq j\leq e} f_j(x)\right\},\\
       M(U) &= \max\{0,f_1(x),f_2(x),f_3(x)\}.
    \end{aligned}
    \]
For an empty block, all states are zero. The same initialization remains valid when the sector coefficients are arbitrary real numbers. For an internal node, we use the following recurrences:
        \begin{equation}\label{eq:generalized_A(U)_recursions}
        \begin{split}
            S_{s,e}(U) &= \max\{S_{s,i}(U_1)+S_{i,e}(U_2)~|~ \text{with } i=s,\ldots,e\},      \\
            L_s(U) &= \max\{L_s(U_1),~ S_{s,i}(U_1) + L_i(U_2)~|~ \text{with } i=s,\ldots,3\}, \\
            M(U) &= \max\{M(U_1),~ R_i(U_1) + L_i(U_2),~ M(U_2)~|~ \text{with } i=1,2,3\},\\
            R_e(U) &= \max\{R_i(U_1) + S_{i,e}(U_2),~ R_e(U_2)~|~ \text{with } i=1,\ldots,e\}.
        \end{split}
    \end{equation}
Two of these computations are depicted in Figure~\ref{fig:M(U)recursion}.

\begin{lemma}\label{lem:generalized_mcs_correct}
For arbitrary real-valued functions $f_1,f_2,f_3$ on $X$, the leaf initialization above together with the recurrences in Equation~\eqref{eq:generalized_A(U)_recursions} computes exactly the quantities $S_{s,e}(U)$, $L_s(U)$, $R_e(U)$, and $M(U)$ defined for every canonical block $U$.
\end{lemma}
\begin{proof}
We argue by induction on the height of the tree. The statement is immediate for leaves. Let $U=U_1U_2$ be an internal node and assume that all states are computed correctly for its two children.

Consider first $S_{s,e}(U)$. Every feasible solution contains both endpoints of $U$, so its consecutive labeled family crosses the boundary between $U_1$ and $U_2$. For some $i\in\{s,\ldots,e\}$, the restriction to $U_1$ uses labels from $s$ through $i$, while the restriction to $U_2$ uses labels from $i$ through $e$; the shared label may be empty on one side. The best contributions are therefore $S_{s,i}(U_1)$ and $S_{i,e}(U_2)$. Maximizing over $i$ gives the first recurrence.

For $L_s(U)$, the optimum either stays inside $U_1$, giving $L_s(U_1)$, or extends into $U_2$. In the latter case, for some $i\in\{s,\ldots,3\}$ the restriction to $U_1$ uses labels from $s$ through $i$ and the remaining suffix in $U_2$ uses labels from $i$ through $3$. These two parts contribute $S_{s,i}(U_1)$ and $L_i(U_2)$. This is exactly the second recurrence. The argument for $R_e(U)$ is symmetric.

Finally, an optimum for $M(U)$ is contained in one child or crosses the boundary. The first two cases contribute $M(U_1)$ and $M(U_2)$. In the crossing case, for some $i\in\{1,2,3\}$ the left restriction ends within the label range $1,\ldots,i$ and the right restriction begins within $i,\ldots,3$, giving $R_i(U_1)+L_i(U_2)$. Taking the maximum over these possibilities proves the recurrence for $M(U)$ and completes the induction.
\end{proof}

It follows that the two-box tree can be built in $O(n)$ time. A change in the values $f_j(x)$ for one point affects only one root-to-leaf path and is propagated in $O(\log n)$ time.
    
\section{Results}\label{sec:RESDISC}

For reference, the central symmetric-difference problem uses the following $18$ sector-activation matrices for the $18$ unlabeled configurations; blank entries are zero:

$$
\begin{bsmallmatrix}1&1&1\\1&~&1\\1&1&1\end{bsmallmatrix}, ~~
\begin{bsmallmatrix}1&1&~\\1&~&1\\1&1&~\end{bsmallmatrix}, ~~
\begin{bsmallmatrix}1&~&~\\1&~&1\\1&~&~\end{bsmallmatrix}, ~~
\begin{bsmallmatrix}~&1&1\\1&~&1\\~&1&1\end{bsmallmatrix}, ~~
\begin{bsmallmatrix}~&1&~\\1&~&1\\~&1&~\end{bsmallmatrix}, ~~
\begin{bsmallmatrix}~&~&1\\1&~&1\\~&~&1\end{bsmallmatrix}, 
$$$$
\begin{bsmallmatrix}1&1&1\\1&~&1\\~&1&~\end{bsmallmatrix}, ~~
\begin{bsmallmatrix}1&1&~\\1&~&1\\~&1&1\end{bsmallmatrix}, ~~
\begin{bsmallmatrix}1&~&~\\1&~&1\\~&~&1\end{bsmallmatrix}, ~~
\begin{bsmallmatrix}~&1&1\\1&~&1\\1&1&~\end{bsmallmatrix}, ~~
\begin{bsmallmatrix}~&1&~\\1&~&1\\1&1&1\end{bsmallmatrix}, ~~
\begin{bsmallmatrix}~&~&1\\1&~&1\\1&~&~\end{bsmallmatrix},
$$$$
\begin{bsmallmatrix}1&1&1\\~&~&~\\~&1&~\end{bsmallmatrix}, ~~
\begin{bsmallmatrix}1&1&~\\~&~&~\\~&1&1\end{bsmallmatrix}, ~~
\begin{bsmallmatrix}1&~&~\\~&~&~\\~&~&1\end{bsmallmatrix}, ~~
\begin{bsmallmatrix}~&1&1\\~&~&~\\1&1&~\end{bsmallmatrix}, ~~
\begin{bsmallmatrix}~&1&~\\~&~&~\\1&1&1\end{bsmallmatrix}, ~~
\begin{bsmallmatrix}~&~&1\\~&~&~\\1&~&~\end{bsmallmatrix}.
$$

These matrices come from the relative placements in Figure~\ref{fig:cases} together with the boundary orders that are distinct for the sweep; rotations and reflections are handled by symmetry.

For completeness, we record the corresponding configuration count, although only the fact that it is constant for fixed $k$ is used in the running-time analysis. Assume first that all box-boundary coordinates are distinct. Ignoring box labels, the $2k$ vertical side positions can be paired into the $k$ $x$-intervals in $(2k-1)!!$ ways, and the horizontal side positions give the same factor. Pairing the resulting horizontal and vertical intervals into rectangles contributes a factor $k!$. Thus the number of unlabeled configurations is $k!((2k-1)!!)^2$, and assigning labels to the rectangles contributes one additional factor $k!$. Here $1!!=1$ and $(2r+1)!!=(2r+1)(2r-1)!!$ for $r\geq1$. Coincident boundary coordinates do not create new incidence patterns, because coincident sides can be separated inside adjacent coordinate gaps without crossing an input point. In particular, for two boxes there are $18$ unlabeled configurations and $36$ labeled configurations. The $18$ matrices above suffice for symmetric difference because that objective is invariant under exchanging the two box labels.

\begin{theorem}\label{thm:two_box_direct}
    The \textsc{Maximum-Weight Two Boxes Symmetric Difference} problem can be solved in $O(n^4\log n)$ time and $O(n)$ space.
\end{theorem}
\begin{proof}
Fix one of the $18$ sector configurations and three of the four horizontal boundaries. Each boundary has only $O(n)$ relevant positions, namely the gaps between consecutive $y$-coordinates, so there are $O(n^3)$ fixed triples. For each triple we build the generalized MCS-tree in $O(n)$ time and sweep the remaining boundary only through the $O(n)$ positions that preserve the fixed strict boundary order. A sweep event changes the data of one point and costs $O(\log n)$ to propagate through the tree. Thus, each fixed triple costs $O(n\log n)$ after initialization, and the total time is $O(n^4\log n)$. The factor $18$ is constant. The tree has $O(n)$ nodes and stores only a constant number of values per node, so it uses $O(n)$ space.
\end{proof}

\subsection{From sector coefficients to Boolean incidence objectives}\label{sec:boolean_objectives}

The proof of Theorem~\ref{thm:two_box_direct} uses the fixed coefficients $a_{i,j}\in\{0,1\}$ determined by the sector-activation matrices of the symmetric-difference objective. However, neither the generalized MCS-tree nor its recurrences require these coefficients to be binary. Once a labeled combinatorial placement of the boxes is fixed, each sector has a fixed incidence pattern, and the coefficient assigned to that sector may be any fixed real value. This observation leads naturally to a more general class of objectives.

Write $[k]=\{1,\ldots,k\}$. Let $\mathcal{F}=(A_1,\ldots,A_k)$ be a labeled family of $k$ axis-aligned boxes. For $p\in P$, define
\begin{equation}\label{eq:incidence_vector}
    \mathbf{z}_{\mathcal{F}}(p)
    =\bigl(\mathds{1}[p\in A_1],\ldots,\mathds{1}[p\in A_k]\bigr)
    \in\{0,1\}^k.
\end{equation}
For a fixed function $h:\{0,1\}^k\to\mathbb{R}$, define the corresponding Boolean incidence objective by
\begin{equation}\label{eq:boolean_objective}
    \Phi_h(\mathcal{F})
    =\sum_{p\in P}\omega(p)\,
      h\!\left(\mathbf{z}_{\mathcal{F}}(p)\right).
\end{equation}
We may assume $h(\mathbf{0})=0$. Indeed, if $\bar h(\mathbf{z})=h(\mathbf{z})-h(\mathbf{0})$, then
\[
    \Phi_h(\mathcal{F})
    =\Phi_{\bar h}(\mathcal{F})
     +h(\mathbf{0})\sum_{p\in P}\omega(p),
\]
and the second term is independent of the selected boxes.

For instance, with three boxes one may set
\[
\begin{array}{c|cccccccc}
\mathbf{z} & 000&100&010&001&110&101&011&111\\ \hline
h(\mathbf{z})&0&4&2&1&3&5&0&6 .
\end{array}
\]
The entries of $h$ are rewards assigned to complete incidence patterns; they are not obtained by adding independent rewards associated with the individual boxes. Thus, points covered only by $A_1$, only by $A_2$, and only by $A_3$ have the same coverage depth but receive values $4$, $2$, and $1$, respectively. No objective depending only on the depth can express this distinction.

An incidence objective is \emph{symmetric} if
\[
 h(z_1,\ldots,z_k)=h(z_{\pi(1)},\ldots,z_{\pi(k)})
\]
for every permutation $\pi$ of $\{1,\ldots,k\}$. These are exactly the objectives that depend only on the Hamming weight $\|\mathbf{z}\|_1=\sum_i z_i$, i.e., on the number of covering boxes. Consequently, the usual coverage-depth objectives form a proper subclass of the present model.

Returning to the two-box direct algorithm, fix a labeled combinatorial placement and let $\boldsymbol{\chi}_{i,j}\in\{0,1\}^2$ denote the incidence vector of the sector in row $i$ and column $j$. The fixed zero-one activation coefficient used for symmetric difference is now replaced by
$ a_{i,j}=h(\boldsymbol{\chi}_{i,j}). $
Accordingly, the vertical-strip functions remain
\[
    f_j(p)=\sum_{i=1}^{3}a_{i,j}\varphi_i(p)\omega(p).
\]
The geometry of the decomposition and all generalized MCS-tree recurrences remain unchanged; only the coefficients assigned to the sectors are replaced according to their incidence patterns.

For example, in the containment configuration $B\subseteq A$, the nine internal sectors have incidence vectors corresponding to $A$ alone except for the central sector, which is covered by both boxes. Three familiar choices are
\[
\begin{aligned}
 h_{\triangle}(z_1,z_2)&=\mathds{1}[z_1\neq z_2],\\
 h_{\cup}(z_1,z_2)&=\mathds{1}[z_1+z_2\geq1],\\
 h_{\cap}(z_1,z_2)&=z_1z_2.
\end{aligned}
\]
They give, respectively,
\[
 \begin{bmatrix}1&1&1\\1&0&1\\1&1&1\end{bmatrix},
 \qquad
 \begin{bmatrix}1&1&1\\1&1&1\\1&1&1\end{bmatrix},
 \qquad
 \begin{bmatrix}0&0&0\\0&1&0\\0&0&0\end{bmatrix}.
\]
Thus, the same geometric decomposition supports both symmetric and label-sensitive objectives; only the sector coefficients change.

Since a general two-box Boolean incidence objective may distinguish the labels $A$ and $B$, we enumerate the $36$ labeled configurations rather than the $18$ unlabeled configurations sufficient for symmetric difference. This affects only the constant factor.

\begin{theorem}\label{thm:two_box_boolean}
Let $h:\{0,1\}^2\rightarrow\mathbb{R}$ be fixed. Two labeled axis-aligned boxes $A$ and $B$ maximizing
\[
    \Phi_h(A,B)
    =\sum_{p\in P}\omega(p)\,
      h\bigl(\mathds{1}[p\in A],\mathds{1}[p\in B]\bigr)
\]
can be found in $O(n^4\log n)$ time and $O(n)$ space.
\end{theorem}
\begin{proof}
Normalize $h$ so that $h(0,0)=0$. In each labeled combinatorial placement, every sector has a fixed incidence vector $\boldsymbol{\chi}_{i,j}\in\{0,1\}^2$. Set $a_{i,j}=h(\boldsymbol{\chi}_{i,j})$ and define $f_j$ as in Equation~\eqref{eq:vertcal_strips}. For fixed horizontal lines, maximizing $\sum_j f_j(X_j)$ is exactly Equation~\eqref{eq:boolean_objective}. Since the generalized MCS-tree recurrences allow arbitrary real values of $f_j$, the sweep analysis of Theorem~\ref{thm:two_box_direct} applies unchanged. The use of $36$ rather than $18$ configurations affects only the constant factor.
\end{proof}

\subsection{A faster algorithm via Weighted Depth}\label{sec:weighted_depth}

The direct algorithm is explicit and uses only linear space, but its running time contains a logarithmic factor. We remove this factor through a parameter-space reduction. The use of rectangle-boundary coordinates as parameter-space variables follows the standard connection between \textsc{Maximum-Weight Box} and \textsc{Weighted Depth}; see Barbay et al.~\cite{BARBAY2014437}. Here two rectangles are represented simultaneously, and the contribution of each input point depends on their joint incidence pattern. For symmetric difference, this joint contribution produces the signed weights $\omega(p),\omega(p),-2\omega(p)$ used below.

An instance of \textsc{Weighted Depth} in dimension $d$ consists of weighted axis-aligned boxes in $\mathbb{R}^d$; weights may be positive or negative. The task is to find a point maximizing the total weight of the boxes that contain it.

For every fixed $d\ge 3$, \textsc{Weighted Depth} on $m$
axis-aligned boxes with arbitrary real weights can be solved in $O(m^{d/2})$ time and linear working space~\cite{BARBAY2014437,Chan2013Klee}. We use this result only for $d=8$ and $d=4k$. The bound follows from Chan's fixed-dimensional Klee-measure framework~\cite{Chan2013Klee} and the weighted formulations used for maximum-weight box problems~\cite{BackursDikkalaTzamos2016,BARBAY2014437}. Negative weights are essential in our reductions.

Let
\[
  x_{\min}=\min_{p\in P}p_x,\quad x_{\max}=\max_{p\in P}p_x,
  \qquad
  y_{\min}=\min_{p\in P}p_y,\quad y_{\max}=\max_{p\in P}p_y.
\]
Choose constants $x_l<x_{\min}\leq x_{\max}<x_r$ and $y_d<y_{\min}\leq y_{\max}<y_u$, and let $I_x=[x_l,x_r]$ and $I_y=[y_d,y_u]$. Since $P\subset I_x\times I_y$, replacing any rectangle $R$ by $R\cap(I_x\times I_y)$ does not change which points of $P$ it contains.
Hence all boundary coordinates may be restricted to these intervals. The non-strict inequalities also cover the one-point case.

Represent two boxes by the eight parameters
\[
 \textbf{q}=(a_x^-,a_x^+,a_y^-,a_y^+,b_x^-,b_x^+,b_y^-,b_y^+)
 \in I_x^2\times I_y^2\times I_x^2\times I_y^2 \subset \mathbb{R}^8.
\]
When the lower endpoints do not exceed the upper endpoints, these coordinates define rectangles $A_{\textbf{q}}$ and $B_{\textbf{q}}$. A reversed interval, for example $a_x^->a_x^+$, satisfies no membership constraint and therefore has empty incidence on $P$; the same incidence is realized by a legitimate rectangle avoiding $P$. Conversely, every legitimate pair can be replaced by a pair whose boundary coordinates lie in $I_x\times I_y$ without changing which input points it contains. Hence the bounded parameter space represents exactly the incidence pairs relevant to the optimization problem.

For a point $p=(p_x,p_y)\in P$, define the following boxes in the parameter space:
\begin{align*}
 Q_p^A={}&[x_l,p_x]\times[p_x,x_r]\times[y_d,p_y]\times[p_y,y_u]
          \times I_x^2\times I_y^2,\\
 Q_p^B={}&I_x^2\times I_y^2\times
          [x_l,p_x]\times[p_x,x_r]\times[y_d,p_y]\times[p_y,y_u],\\
 Q_p^{AB}={}&Q_p^A\cap Q_p^B.
\end{align*}
These are axis-aligned boxes in $\mathbb{R}^8$, and their meaning is immediate from the construction:
\[
 \textbf{q}\in Q_p^A \Longleftrightarrow p\in A_{\textbf{q}},
 \qquad
 \textbf{q}\in Q_p^B \Longleftrightarrow p\in B_{\textbf{q}},
 \qquad
 \textbf{q}\in Q_p^{AB} \Longleftrightarrow p\in A_{\textbf{q}}\cap B_{\textbf{q}}.
\]
Thus each input point contributes exactly three boxes to the Weighted Depth instance.

\begin{theorem}\label{thm:two_box_weighted_depth}
The \textsc{Maximum-Weight Two Boxes Symmetric Difference} problem can be solved in $O(n^4)$ time and $O(n)$ space.
\end{theorem}
\begin{proof}
For each $p\in P$, give $Q_p^A$ and $Q_p^B$ weight $\omega(p)$, and give $Q_p^{AB}$ weight $-2\omega(p)$. The resulting instance has $m=3n$ boxes in $\mathbb{R}^8$.

For every parameter point $\textbf{q}$, the defining equivalences of $Q_p^A,Q_p^B,Q_p^{AB}$ give
\begin{align*}
 \operatorname{WD}(\textbf{q})
 &=\sum_{p\in P}\omega(p)\Bigl(
      \mathds{1}[p\in A_{\textbf{q}}]
     +\mathds{1}[p\in B_{\textbf{q}}]
     -2\mathds{1}[p\in A_{\textbf{q}}\cap B_{\textbf{q}}]
   \Bigr)\\
 &=\sum_{p\in P}\omega(p)\,\mathds{1}[p\in A_{\textbf{q}}\triangle B_{\textbf{q}}]\\
 &=\omega(A_{\textbf{q}}\triangle B_{\textbf{q}}),
\end{align*}
where the second equality uses
$
 \mathds{1}[p\in A\triangle B]
 =\mathds{1}[p\in A]+\mathds{1}[p\in B]
  -2\mathds{1}[p\in A\cap B].
$
Since the bounded parameter space represents exactly the realizable incidence pairs on $P$, we have
\[
    \max_{\textbf{q}}\operatorname{WD}(\textbf{q})
    =\max_{A,B}\omega(A\triangle B).
\]

All parameter boxes $Q_p^A$, $Q_p^B$, and $Q_p^{AB}$ are contained in the
bounded parameter space
$ I_x^2\times I_y^2\times I_x^2\times I_y^2, $ so the weighted depth is zero outside this domain. On the other hand, the original problem always has a feasible solution of value zero: choosing $A=B$ gives $A\triangle B=\emptyset$. Hence its optimum is nonnegative. Therefore the zero value outside the bounded domain cannot improve on the original optimum, and we may maximize Weighted Depth over all of $\mathbb{R}^8$.

Now $d=8$ and $m=3n$, so the fixed-dimensional Weighted Depth algorithm runs in
\[
 O(m^{d/2})=O((3n)^4)=O(n^4)
\]
time and uses $O(m)=O(n)$ working space.
\end{proof}

The same construction handles an arbitrary fixed two-box Boolean incidence objective. After normalizing $h(0,0)=0$, set
\[
    c_{\{1\}}=h(1,0),\qquad
    c_{\{2\}}=h(0,1),\qquad
    c_{\{1,2\}}=h(1,1)-h(1,0)-h(0,1).
\]
For $z_1,z_2\in\{0,1\}$ we then have the unique multilinear identity
\begin{equation}\label{eq:two_box_boolean_mobius}
 h(z_1,z_2)=c_{\{1\}}z_1+c_{\{2\}}z_2+c_{\{1,2\}}z_1z_2.
\end{equation}
Therefore $Q_p^A,Q_p^B,Q_p^{AB}$ receive weights $c_{\{1\}}\omega(p)$, $c_{\{2\}}\omega(p)$, and $c_{\{1,2\}}\omega(p)$, respectively. Symmetric difference is recovered from $c_{\{1\}}=c_{\{2\}}=1$ and $c_{\{1,2\}}=-2$.

\begin{corollary}\label{cor:two_box_boolean_n4}
Let $h:\{0,1\}^2\rightarrow\mathbb{R}$ be fixed. Two labeled axis-aligned boxes maximizing $\Phi_h$ can be found in $O(n^4)$ time and $O(n)$ space.
\end{corollary}
\begin{proof}
Normalize $h(0,0)=0$ and use Equation~\eqref{eq:two_box_boolean_mobius}. The weighted depth at a parameter point $\textbf{q}$ is exactly $\Phi_h(A_{\textbf{q}},B_{\textbf{q}})$. 
Since both boxes may be chosen disjoint from $P$ and $h(0,0)=0$, the original objective has value zero for some feasible solution and hence its optimum is nonnegative. Therefore the zero weighted depth outside the bounded parameter domain cannot improve the optimum. The dimension is eight and the instance has $3n$ parameter boxes, giving $O(n^4)$ time and $O(n)$ working space.
\end{proof}

\paragraph{Relation with general Weighted Depth.}
The coefficient pattern $(w,w,-2w)$ alone should not be expected to yield a better exponent than general Weighted Depth. Let $\mathcal D\subset\mathbb{R}^4$ denote the bounded parameter domain of one planar rectangle, and consider an arbitrary weighted box $B_i\times B_j\subseteq\mathcal D\times\mathcal D$ of weight $\lambda$. For a parameter point $(x,y)$, write
\[
    a=\mathds{1}[x\in B_i],\qquad b=\mathds{1}[y\in B_j],
\]
and define
\[
    g(a,b)=a+b-2ab.
\]
A term $w\,g(a,b)$ is exactly the contribution of the three weighted boxes
\[
    B_i\times\mathcal D,\qquad \mathcal D\times B_j,\qquad B_i\times B_j
\]
with weights $w,w,-2w$, respectively. Moreover,
\[
    g(a,1)=1-a,\qquad g(1,b)=1-b,
\]
so
\[
    \lambda ab
    =\lambda-\frac{\lambda}{2}
      \bigl(g(a,b)+g(a,1)+g(1,b)\bigr).
\]
Thus one arbitrary Weighted Depth box can be represented, up to an additive constant, by three structured terms having the same coefficient pattern, while every such structured term is already a sum of three weighted boxes. Consequently, when arbitrary factor boxes are allowed, this signed coefficient structure is linearly equivalent to general Weighted Depth up to an additive constant.

The instances arising from our geometric problem are more structured: for each input point, the constraints on the two rectangle factors are induced by that same point. Hence the equivalence above does not rule out a faster algorithm that exploits this additional point-induced structure.

We next extend both algorithms from two boxes to any fixed number $k$.

\subsection{Maximum-Weight $k$ Boxes}\label{sec:k_boxes}
For fixed $k$, brute force considers $O(n^{4k})$ combinatorially distinct placements. Evaluating every placement on all $n$ input points takes $O(n^{4k+1})$ time and $O(1)$ auxiliary space. The direct framework improves this bound by using the $2k$ horizontal boundaries to form $2k-1$ relevant strips and $2k-1$ vertical-block functions. We fix $2k-1$ horizontal boundaries and sweep the last one, which gives $O(n^{2k}\log n)$ time for each fixed labeled combinatorial configuration.

Let $q=2k-1$. The generalized MCS-tree now chooses $q$ consecutive labeled blocks of the $x$-ordered sequence $X$. At each node we store one $M$ value, $q$ values of type $L$, $q$ of type $R$, and
\[
    \frac{q(q+1)}{2}=\frac{(2k-1)(2k)}{2}
\]
values of type $S$. There are therefore $O(k^2)$ values per node and $O(k^2n)$ values in the tree. Each state takes a maximum over at most $O(k)$ split indices, so a merge costs $O(k^3)$. Building the tree takes $O(k^3n)$ time and updating one point takes $O(k^3\log n)$ time. When $k$ is fixed, these become $O(n)$ and $O(\log n)$, respectively.

The merge rules are the same as in Equation~\eqref{eq:generalized_A(U)_recursions}, with $3$ replaced by $q$. For $U=U_1U_2$,
\begin{equation}\label{eq:k_generalized_mcs}
\begin{split}
S_{s,e}(U) &= \max_{s\le i\le e}\{S_{s,i}(U_1)+S_{i,e}(U_2)\},\\
L_s(U) &= \max\Bigl\{L_s(U_1),~\max_{s\le i\le q}\{S_{s,i}(U_1)+L_i(U_2)\}\Bigr\},\\
M(U) &= \max\Bigl\{M(U_1),M(U_2),~\max_{1\le i\le q}\{R_i(U_1)+L_i(U_2)\}\Bigr\},\\
R_e(U) &= \max\Bigl\{R_e(U_2),~\max_{1\le i\le e}\{R_i(U_1)+S_{i,e}(U_2)\}\Bigr\}.
\end{split}
\end{equation}
The state meanings and leaf initialization are unchanged, with indices in $\{1,\ldots,q\}$. The proof of Lemma~\ref{lem:generalized_mcs_correct} carries over verbatim: an optimal consecutive labeled family is contained in one child or crosses the child boundary at a single label split. Thus, the merge cost is $O(k^3)$, which is constant with respect to $n$ when $k$ is fixed.

An arbitrary fixed Boolean incidence function requires no new data structure. In a fixed labeled configuration, let $\boldsymbol{\chi}_{i,j}\in\{0,1\}^k$ be the incidence vector of a sector and set
$ a_{i,j}=h(\boldsymbol{\chi}_{i,j}).
$ The $2k-1$ functions $f_j$ are defined as before.

For completeness, the same direct construction also gives a fixed-$k$ bound for arbitrary Boolean incidence objectives. Fix a labeled relative order of all horizontal and vertical sides. The incidence vector of every sector is then fixed, so after setting $a_{i,j}=h(\boldsymbol{\chi}_{i,j})$ we obtain the one-dimensional problem on $2k-1$ labeled blocks. There are $O(n^{2k-1})$ choices for the fixed horizontal boundaries, and for each choice the remaining boundary is swept through $O(n)$ relevant positions. Each event costs $O(k^3\log n)$ and the initial tree costs $O(k^3n)$. Because the number of labeled combinatorial configurations depends only on the fixed parameter $k$, the resulting running time is $O(n^{2k}\log n)$ and the space usage is $O(n)$. If $h$ is symmetric under permutations of the box labels, unlabeled configurations suffice, reducing only the hidden constant.

For fixed $k$, the direct construction above gives $O(n^{2k}\log n)$ time. To remove the logarithmic factor, we again move to parameter space. The algebraic ingredient is classical: every function on the Boolean cube has a unique multilinear representation, equivalently obtained by M\"obius inversion on the subset lattice. Applying this representation to the incidence vector of a point converts an arbitrary fixed Boolean incidence objective into a weighted sum of simultaneous membership events.

\begin{lemma}\label{lem:boolean_mobius}
Let $h:\{0,1\}^k\to\mathbb{R}$ and assume $h(\mathbf{0})=0$. There are unique coefficients $c_S$, one for every nonempty $S\subseteq[k]$, such that
\begin{equation}\label{eq:boolean_multilinear_expansion}
 h(z_1,\ldots,z_k)=\sum_{\emptyset\neq S\subseteq[k]} c_S\prod_{i\in S}z_i
 \qquad\text{for all }(z_1,\ldots,z_k)\in\{0,1\}^k.
\end{equation}
They are given by
\begin{equation}\label{eq:boolean_mobius_coefficients}
 c_S=\sum_{T\subseteq S}(-1)^{|S|-|T|}\,h(\mathbf{1}_T),
\end{equation}
where $\mathbf{1}_T$ is the characteristic vector of $T$.
\end{lemma}
\begin{proof}
This is M\"obius inversion on the subset lattice. For a vector $\mathbf{z}=\mathbf{1}_U$, the monomial $\prod_{i\in S}z_i$ equals $1$ exactly when $S\subseteq U$. Hence Equation~\eqref{eq:boolean_multilinear_expansion} is equivalent to
$ h(\mathbf{1}_U)=\sum_{S\subseteq U}c_S. $
Applying M\"obius inversion to this relation gives Equation~\eqref{eq:boolean_mobius_coefficients}. Since $h(\mathbf{0})=0$, the coefficient of the empty set is zero.
\end{proof}

For the three-box example introduced in Section~\ref{sec:boolean_objectives},
M\"obius inversion gives
\[
\begin{aligned}
 h(z_1,z_2,z_3)
 ={}&4z_1+2z_2+z_3-3z_1z_2-3z_2z_3+5z_1z_2z_3.
\end{aligned}
\]
For instance, at $(1,1,1)$ the value is
$4+2+1-3-3+5=6$, whereas at $(1,1,0)$ it is $4+2-3=3$, exactly as prescribed by the incidence table. This illustrates how the multilinear coefficients correct the multiple counting of simultaneous memberships.

Let $\mathcal{F}=(A_1,\ldots,A_k)$. For every $p\in P$,
\[
 \prod_{i\in S}z_i(p)=\mathds{1}\!\left[p\in\bigcap_{i\in S}A_i\right].
\]
Substituting this identity into Equation~\eqref{eq:boolean_multilinear_expansion} yields
\begin{equation}\label{eq:boolean_intersection_expansion}
 \Phi_h(\mathcal{F})=\sum_{\emptyset\neq S\subseteq[k]} c_S\,\omega\!\left(\bigcap_{i\in S}A_i\right),
\end{equation}
where $\omega(Q)=\sum_{p\in P\cap Q}\omega(p)$.

If $h$ is symmetric, then $c_S$ depends only on $|S|$. Writing this common coefficient as $c_s$, we obtain
\[
 h(\mathbf{z})=\sum_{s=1}^k c_s\sum_{\substack{S\subseteq[k]\\|S|=s}}\prod_{i\in S}z_i
 =\sum_{s=1}^k c_s\binom{\|\mathbf{z}\|_1}{s}.
\]
Thus, the familiar finite Newton/binomial expansion is precisely the permutation-invariant special case of the Boolean-lattice M\"obius expansion.

\begin{theorem}\label{thm:k_box_weighted_depth}
Let $k$ be fixed and let $h:\{0,1\}^k\rightarrow\mathbb{R}$ be fixed. The problem of selecting $k$ labeled axis-aligned boxes maximizing $\Phi_h$ can be solved in $O(n^{2k})$ time and $O(n)$ space, where the constants depend on $k$.
\end{theorem}
\begin{proof}
Normalize $h$ so that $h(\mathbf{0})=0$, and let $\{c_S:\emptyset\neq S\subseteq[k]\}$ be the coefficients from Lemma~\ref{lem:boolean_mobius}. Represent each planar box $A_i$ by its four boundary coordinates
$(a_{i,x}^-,a_{i,x}^+,a_{i,y}^-,a_{i,y}^+), $ so the full parameter space has dimension $4k$. As before, all coordinates may be restricted to fixed bounded intervals containing the input.

For $p\in P$ and a nonempty $S\subseteq[k]$, let $Q_{p,S}$ contain exactly those parameter vectors for which $p\in A_i$ for every $i\in S$. For each $i\in S$, this imposes
\[
 a_{i,x}^-\leq p_x\leq a_{i,x}^+,
 \qquad
 a_{i,y}^-\leq p_y\leq a_{i,y}^+,
\]
and there is no restriction for $i\notin S$. Hence $Q_{p,S}$ is an axis-aligned box in $\mathbb{R}^{4k}$.

Give $Q_{p,S}$ weight $c_S\omega(p)$. By Equation~\eqref{eq:boolean_intersection_expansion}, the weighted depth of a parameter vector $\textbf{q}$ is exactly $\Phi_h(A_1,\ldots,A_k)$. Reversed intervals can be treated as empty and replaced by point-avoiding rectangles, exactly as in the two-box case. The weighted depth is zero outside the bounded domain. Since all selected boxes may avoid $P$ and $h(\mathbf{0})=0$, there is always a feasible solution of value zero. Hence the optimum of the original problem is nonnegative, and unrestricted Weighted Depth has the same optimum. The construction creates at most
\[
 m=(2^k-1)n=O(n)
\]
parameter boxes for fixed $k$. Since the dimension is $d=4k$, Weighted Depth runs in
\[
 O(m^{d/2})=O\!\left(m^{2k}\right)=O(n^{2k}).
\]
The working space is $O(m)=O(n)$ for fixed $k$. A maximizing parameter point encodes the $4k$ coordinates of an optimal labeled family of rectangles.
\end{proof}

For the symmetric-difference objective on $k$ boxes, take
\[
 h_{\triangle}(z_1,\ldots,z_k)=\mathds{1}\!\left[\sum_{i=1}^k z_i \text{ is odd}\right].
\]
Theorem~\ref{thm:k_box_weighted_depth} therefore gives $O(n^{2k})$ time for every fixed $k$. The direct algorithm and the Weighted Depth reduction hide different constants in $k$: the direct method depends on the finite number of boundary-order configurations, while the parameter-space reduction depends on the $2^k-1$ boxes generated per input point and on the dimension $4k$. Since $k$ is fixed, neither dependence changes the exponent of $n$.




\section{Applications and Discussion}\label{sec:applications}

\subsection{Consequences and other variants}\label{sec:variants}

Several familiar objectives are immediate special cases of the Boolean incidence model. For union, use
\[
 h_{\cup}(z_1,\ldots,z_k)=\mathds{1}\!\left[\sum_{i=1}^k z_i\geq1\right].
\]

\begin{corollary}\label{cor:union}
The \textsc{Maximum-Weight Two Boxes Union} problem can be solved in $O(n^4)$ time.
\end{corollary}

For intersection, use
\[
 h_{\cap}(z_1,\ldots,z_k)=\prod_{i=1}^k z_i.
\]
For a fixed $t$, exact-depth and threshold objectives are obtained from
\[
 h_{=t}(\mathbf{z})=\mathds{1}[\|\mathbf{z}\|_1=t],
 \qquad
 h_{\geq t}(\mathbf{z})=\mathds{1}[\|\mathbf{z}\|_1\geq t].
\]

\begin{corollary}\label{cor:depth_special_cases}
For fixed $k$ and fixed $t\in\{1,\ldots,k\}$, the following problems can all be solved in $O(n^{2k})$ time: maximum-weight union of $k$ boxes, maximum-weight intersection of $k$ boxes, maximum weight covered by exactly $t$ boxes, and maximum weight covered by at least $t$ boxes.
\end{corollary}

These are uniform bounds supplied by the Boolean incidence framework, not necessarily the best bounds for each special case. For instance, the intersection of $k$ axis-aligned boxes is itself an axis-aligned box, so maximum-weight intersection reduces to the classical planar \textsc{Maximum-Weight Box} problem and can be solved in $O(n^2)$ time~\cite{BARBAY2014437}.

\paragraph{Disjoint boxes.}
If $A$ and $B$ are disjoint, then
\[
    A\triangle B=A\cup B
    \qquad\text{and}\qquad
    \omega(A\triangle B)=\omega(A)+\omega(B).
\]
A simple separator argument gives an $O(n^3)$ algorithm for arbitrary real weights: two disjoint axis-aligned rectangles are separated by a vertical or horizontal line, and for each of the $O(n)$ candidate separators the best rectangle on either side is obtained from two independent planar \textsc{Maximum-Weight Box} instances. In the classical bichromatic discrepancy setting, Dobkin et al.~\cite{DOBKIN1996453} already gave the stronger $O(n^2\log n)$ bound for the union of two disjoint axis-aligned rectangles by processing all separators within a common sweep.

\paragraph{The two-box objective as a three-parameter family}
After normalizing $h(0,0)=0$, an arbitrary two-box incidence objective is described by three values
\[
    \alpha=h(1,0),\qquad
    \beta=h(0,1),\qquad
    \gamma=h(1,1).
\]
Let $W_{10}$, $W_{01}$, and $W_{11}$ denote the total weights of points with the corresponding nonzero incidence patterns. Then
\begin{equation}\label{eq:three_parameter_boolean}
    \Phi_h(A,B)=\alpha W_{10}(A,B)+\beta W_{01}(A,B)+\gamma W_{11}(A,B).
\end{equation}
Thus, the new model can distinguish the two single-coverage regions. Familiar objectives correspond to
\[
\begin{array}{c|c|c|c}
\text{Objective} & \alpha & \beta & \gamma\\ \hline
\text{Symmetric difference} & 1 & 1 & 0\\
\text{Union} & 1 & 1 & 1\\
\text{Intersection} & 0 & 0 & 1\\
\text{Asymmetric single coverage with overlap penalty} & \rho_A & \rho_B & -\lambda
\end{array}
\]
for fixed real constants $\rho_A,\rho_B,\lambda$. Symmetric two-box objectives are exactly the subfamily satisfying $\alpha=\beta$; the former two-parameter coverage-depth model is therefore recovered as a proper special case.

\paragraph{Rectilinear annuli and cross-shaped regions}
Containment gives another useful interpretation. If $B\subseteq A$, then
\[
    A\triangle B=A\setminus B
    \qquad\text{and}\qquad
    \omega(A\triangle B)=\omega(A)-\omega(B),
\]
where $\omega(\mathcal{R})=\sum_{p\in P\cap \mathcal{R}}\omega(p)$. Hence restricting the algorithm to containment solves the maximum-weight axis-aligned rectilinear-annulus problem, without requiring the inner and outer rectangles to be concentric. Its symmetric-difference activation matrix is $ \begin{bsmallmatrix}1&1&1\\1&~&1\\1&1&1\end{bsmallmatrix}. $
Likewise, $ \begin{bsmallmatrix}~&1&~\\1&1&1\\~&1&~\end{bsmallmatrix}
$ represents a cross-shaped axis-aligned polygon. Individual sector patterns can therefore describe geometric optimization problems of interest.

\section{Conclusions}\label{sec:Conclu}

We studied weighted optimization with axis-aligned rectangles under Boolean incidence objectives, beginning with Maximum-Weight Two Boxes Symmetric Difference. A generalized MCS-tree gives an explicit \(O(n^4\log n)\)-time, \(O(n)\)-space algorithm, while a reduction to Weighted Depth in \(\mathbb{R}^8\) improves the running time to \(O(n^4)\).

More generally, for fixed \(k\) and any fixed \(h:\{0,1\}^k\to\mathbb{R}\), Möbius inversion expresses the objective as a weighted combination of subfamily intersections, yielding an \(O(n^{2k})\)-time algorithm via Weighted Depth in \(\mathbb{R}^{4k}\). This framework includes symmetric difference, union, intersection, exact-depth and threshold coverage, and label-sensitive overlap rewards or penalties.

Finally, the signed coefficient structure \((w,w,-2w)\) induced by symmetric difference does not by itself permit a better exponent: with arbitrary factor boxes and signed weights, it is linearly equivalent to general Weighted Depth up to an additive constant. Whether the additional point-induced structure of our original problem can be exploited to improve the \(O(n^4)\) bound remains open.

\small 
\bibliographystyle{abbrv} 
\bibliography{refs_h}
\end{document}